\documentclass[letterpaper, 10 pt, conference]{IEEEtran}
\usepackage{amsmath, amsfonts}
\usepackage{amssymb,mathrsfs}
\usepackage{epsf, epsfig, graphicx}
\usepackage{graphicx,epsf,mathrsfs,dsfont}
\usepackage{amstext}
\usepackage{color,cite,times,latexsym,ifthen, enumerate}
\usepackage[latin9]{inputenc}
\usepackage{verbatim}
\usepackage{amsmath}
\usepackage{amssymb}
\usepackage{graphicx}
\usepackage{stmaryrd}
\usepackage[final]{pdfpages}
\usepackage{amsthm}
\makeatletter
\usepackage{float}
\usepackage{mathrsfs}
\usepackage{dsfont}
\usepackage{psfrag}
\usepackage{url}
\usepackage{bbm}
\usepackage{multirow}
\usepackage{booktabs}
\usepackage{amsfonts}
\usepackage{epsf}
\usepackage{epsfig}\usepackage{graphicx}\usepackage{color}\usepackage{cite}
\usepackage{times}\usepackage{latexsym}\usepackage{ifthen}\usepackage{psfrag}
\usepackage{subfigure}

\usepackage[bookmarks]{hyperref}
\hypersetup{
	colorlinks,
	citecolor=black,
	filecolor=black,
	linkcolor=black,
	urlcolor=black
}
\usepackage{hypcap}

\renewcommand{\P}{\mathbb{P}}

\newtheorem{definition}{Definition}

\newtheorem{remark}{Remark}

\usepackage{xparse}
\DeclareDocumentCommand{\repairSpc}{m O{} O{}}
{
		\cS_{#1}^{#2}
		\IfNoValueTF{#3}%
		{}%
		{\@ifmtarg{#3}{}{\left[#3\right]}}%
}
\DeclareDocumentCommand{\cmpSpc}{m O{} O{}}
{
	\cT_{#1}^{#2}
	\IfNoValueTF{#3}%
	{}%
	{\@ifmtarg{#3}{}{\left[#3\right]}}%
}
\DeclareDocumentCommand{\cmpVec}{m O{} O{}}
{
	t_{#1}^{#2}
	\IfNoValueTF{#3}%
	{}%
	{\@ifmtarg{#3}{}{\left[#3\right]}}%
}
\DeclareDocumentCommand{\nodeCmp}{m O{} O{}}
{
	\cU_{#1}^{#2}
	\IfNoValueTF{#3}%
	{}%
	{\@ifmtarg{#3}{}{\left[#3\right]}}%
}
\DeclareDocumentCommand{\ncmpVec}{m O{} O{}}
{
	u_{#1}^{#2}
	\IfNoValueTF{#3}%
	{}%
	{\@ifmtarg{#3}{}{\left[#3\right]}}%
}
\renewcommand{\vec}[1]{\mbox{\boldmath$#1$}}

\renewcommand{\dim}[1]{\mathsf{dim}\left(#1\right)}
\newcommand{\set}[1]{\left\lbrace #1 \right\rbrace}
\newcommand{\size}[1]{ \left| #1 \right|}
\newcommand{\intv}[1]{I_{#1}}

\newcommand{\node}[1]{\cW_{#1}}
\newcommand{\nSym}[1]{w_\star{\scriptstyle ({#1}) } }
\newcommand{\sVec}[2]{s_{#2}{\scriptstyle(#1)}}

\let\textquotedbl="

\newcommand{\sspan}[1]{\left\langle #1 \right\rangle}

\newcommand{\cA}{\mathcal{A}}
\newcommand{\cB}{\mathcal{B}}
\newcommand{\cC}{\mathcal{C}}

\newcommand{\cF}{\mathcal{F}}

\newcommand{\cS}{\mathcal{S}}
\newcommand{\cT}{\mathcal{T}}
\newcommand{\cU}{\mathcal{U}}

\newcommand{\cW}{\mathcal{W}}

\newcommand{\dm}[1]{\mathsf{dim}\left( #1 \right)}

\newcommand{\sB}{\mathscr{B}}

\newcommand{\bbF}{\mathbb{F}}

\newtheorem{prop}{Proposition}
\newtheorem{cor}{Corollary}
\newtheorem{thm}{Theorem}

\title{\vspace{0.3in} A Probabilistic Approach Towards Exact-Repair Regeneration Codes}
\author{
\IEEEauthorblockN{Mehran Elyasi}
\IEEEauthorblockA{Department of ECE \\ University of Minnesota \\ melyasi@umn.edu}
\and
\IEEEauthorblockN{Soheil Mohajer}
\IEEEauthorblockA{Department of ECE \\ University of Minnesota\\ soheil@umn.edu}
}

\usepackage[paper=letterpaper,top=1in,bottom=.75in,right=0.75in,left=0.75in]{geometry}
\newgeometry{top=.75in,bottom=0.75in,right=0.75in,left=0.75in}

\begin{document}

\maketitle
\begin{abstract}
Regeneration codes  with exact-repair property for distributed storage systems is studied in this paper. For exact-repair problem, the achievable points of $(\alpha,\beta)$ tradeoff match with the outer bound only for minimum storage regenerating (MSR), minimum bandwidth regenerating (MBR), and some specific values of $n$, $k$, and $d$. Such tradeoff is characterized in this work for general $(n,k,k)$, (i.e., $k=d$) for some range of per-node storage ($\alpha$) and repair-bandwidth ($\beta$). Rather than explicit  code construction, achievability of these tradeoff points is shown by proving \emph{existence} of exact-repair regeneration codes for any $(n,k,k)$.   
More precisely, it is shown that an $(n,k,k)$ system can be extended by adding a new node, which is randomly picked from some ensemble, and it is proved that, with high probability, the existing nodes together with the newly added one maintain properties of exact-repair regeneration codes. The new achievable region improves upon the existing code constructions. In particular, this result provides a complete tradeoff characterization for an $(n,3,3)$ distributed storage system for any value of $n$.
\end{abstract}

\section{Introduction}
Distributed storage systems (DSS) are widely being used to provide reliability to storage technologies. Regeneration codes play a central role to manage data collection as well as system maintenance in DSS . An $(n,k,d)$ regeneration-code encodes a file comprised of $F$ symbols from a finite field $\bbF_q$ into $n$ segments (nodes) $W_1, W_2, \dots, W_n$, each including $\alpha$ symbols. Each data collector is able to recover the entire file by accessing any subset of nodes of size at least $k$. Moreover, whenever a node fails, it can be repaired by accessing $d$ remaining nodes and downloading $\beta$ symbols from each. 
The repair process can be performed in the functional or exact sense. In functional repair, a failed node will be replaced by another one, so that the resulting nodes  maintain the data-recovery as well as node-repair properties. In exact repair (ER), however, the content of a failed node will be exactly replicated by the helpers.  

It turns out that there is a fundamental tradeoff between minimum required values of 
$\alpha$ and $\beta$  to store a given amount of data.  Such tradeoff is derived for function-repair regeneration codes by Dimakis et al. \cite{dimakis2010network}, which is given by 
\[
F\leq \sum\nolimits_{i=0}^{k-1}\min(\alpha, (d-i)\beta).
\]
In practical applications, however, exact-repair is an appealing property, specially when it is desirable that the stored contents remain intact over time.
In contrast to functional-repair, characterizing the optimal tradeoff between per-node storage and repair-bandwidth is a widely open problem for exact-repair regeneration codes.

This tradeoff is only characterized for very special cases. In particular, the optimum tradeoff of a $(4,3,3)$ system is characterized in \cite{ tian2013rate} using a computer-aided approach,  where it was shown that functional and exact repair tradeoffs are not identical. Moreover, the tradeoff is partially characterized in \cite{mohajer2015exact} for a $(5,4,4)$ system, which is extended to a complete characterization in \cite{tian2015note}. For general $(n,k,d)$-DSS, families of outer bounds are developed,  (e.g. \cite{sasidharan2014improved, mohajer2015new}), which can only partially characterize the tradeoff. Recently, the ER  tradeoff is characterized for $(n=k+1, k, d=k)$ systems, independently in \cite{elyasi2015linear, prakash2015storage, duursma2015shortened}, under the assumption of employing linear codes.

On the other hand, efficient ER regeneration codes are introduced for the minimum bandwidth regeneration (MBR) and minimum storage regeneration (MSR) points \cite{rashmi2011optimal, suh2010existence, cadambe2010distributed}. In spite of many other interesting proposed  code construction (e.g. \cite{tian2014layered, goparaju2014new, tamo2013zigzag}), our knowledge about scalability of the system is very limited, for the interior of the tradeoff. 

In this work, we study the scalability problem, and show that for any $n\geq k+1$, the optimum tradeoff of an $(n,k,k)$ ER system matches for some range of $(\alpha,\beta)$, with that of a $(k+1,k,k)$ system. Unlike the standard approach in the literature where achievability of the tradeoff is shown by providing explicit code construction, we use a novel approach, based on random coding. Similar to Shannon's random coding argument, we show that an existing code can be extended by appending new randomly generated nodes, which with high probability maintain the exact-repair property. To the best of our knowledge, this is the first result which suggests that the ER tradeoff may only depend on $(k,d,\alpha,\beta)$ and not the number of nodes $n$ for the interior of the region (similar to the functional repair case). In particular, our approach yields in a complete ER tradeoff characterization for an $(n,3,3)$ DSS. 

The rest of this paper is organized as follows. We first formally present the results in Section~\ref{sec:result}. In Section~\ref{sec:code-struct} we characterize a set of conditions a newly added node should satisfy in order to maintain ER property, and in Section~\ref{sec:prob} we show that w.h.p. a randomly chosen node satisfies these properties. Built on the tools developed in Sections~\ref{sec:code-struct} and ~\ref{sec:prob}, the proof of the main results are presented in Section~\ref{sec:proof}. 

\section{The Model and Main Result}
\label{sec:result}
An exact-repair regenerating distributed storage system with parameters $(n, k, d)$ and $(\alpha, \beta)$ consists of $n$ storage nodes, each with storage capacity $\alpha$ symbols.  A file including $F$ symbols from some finite field $\bbF_q$ is encoded into $n$ pieces of information, in distributed manner, and each piece is stored on one of the storage nodes. 
We associate each node with an index in $\intv n\triangleq\{1,2,\dots,n\}$, and denote by $W_i$ the content of node $i$, for $i\in \intv n$. \emph{Data recovery} property implies that the original file can be recovered from any subset of nodes $\cA\subseteq \intv n$, provided that $|\cA|\geq k$. Moreover, if any node $x\in \intv n$ fails, its content $\cW_x$ can be duplicated by receiving (at most) $\beta$ symbols from each node in a set $\cA$ (called \textit{helper nodes}), for every $\cA \subseteq  \intv n \setminus \set x $ with $|\cA|\geq d$. For a given file size $F$,  characterizing the optimum tradeoff between $\alpha$ and $\beta$ satisfying the aforementioned constraints is a challenging open problem for general $(n,k,d)$-DSS.

Recently, the optimum tradeoff between the per-node storage and repair-bandwidth of ER regeneration codes is characterized in \cite{elyasi2015linear,prakash2015storage, duursma2015shortened} for an $(n=k+1,k,d=k)$-DSS, under the limitation of employing linear codes.  It is shown that the optimum tradeoff for a $(k+1,k,k)$-ER linear DSS is a piecewise-linear function,  with $k$ corner point. The $m$-th corner points of this region is given by
	 \begin{align}
	 \bar{ \alpha} = \frac{m+1}{m(k+1)} \qquad \textrm{and} \qquad
	 \bar{\beta} =  \frac{m+1}{k(k+1)},
	 \label{eq:corner}
	 \end{align}
for $m=1,2,\dots,k$. 

In particular, it was shown that a class of codes  proposed in \cite{tian2014layered} and \cite{goparaju2014new} achieve this optimum tradeoff for any $(k+1,k,k)$. Moreover, the linearity constraint is relaxed for $k=3$  and $k=4$ in \cite{tian2013rate} and \cite{tian2015note}, respectively, where the optimality of the tradeoff is proved using general information-theoretic arguments. This tradeoff is shown  for $k=4$ in Figure~\ref{fig:region}. 

\begin{figure}[t]
\centering
\includegraphics[width=0.35\textwidth]{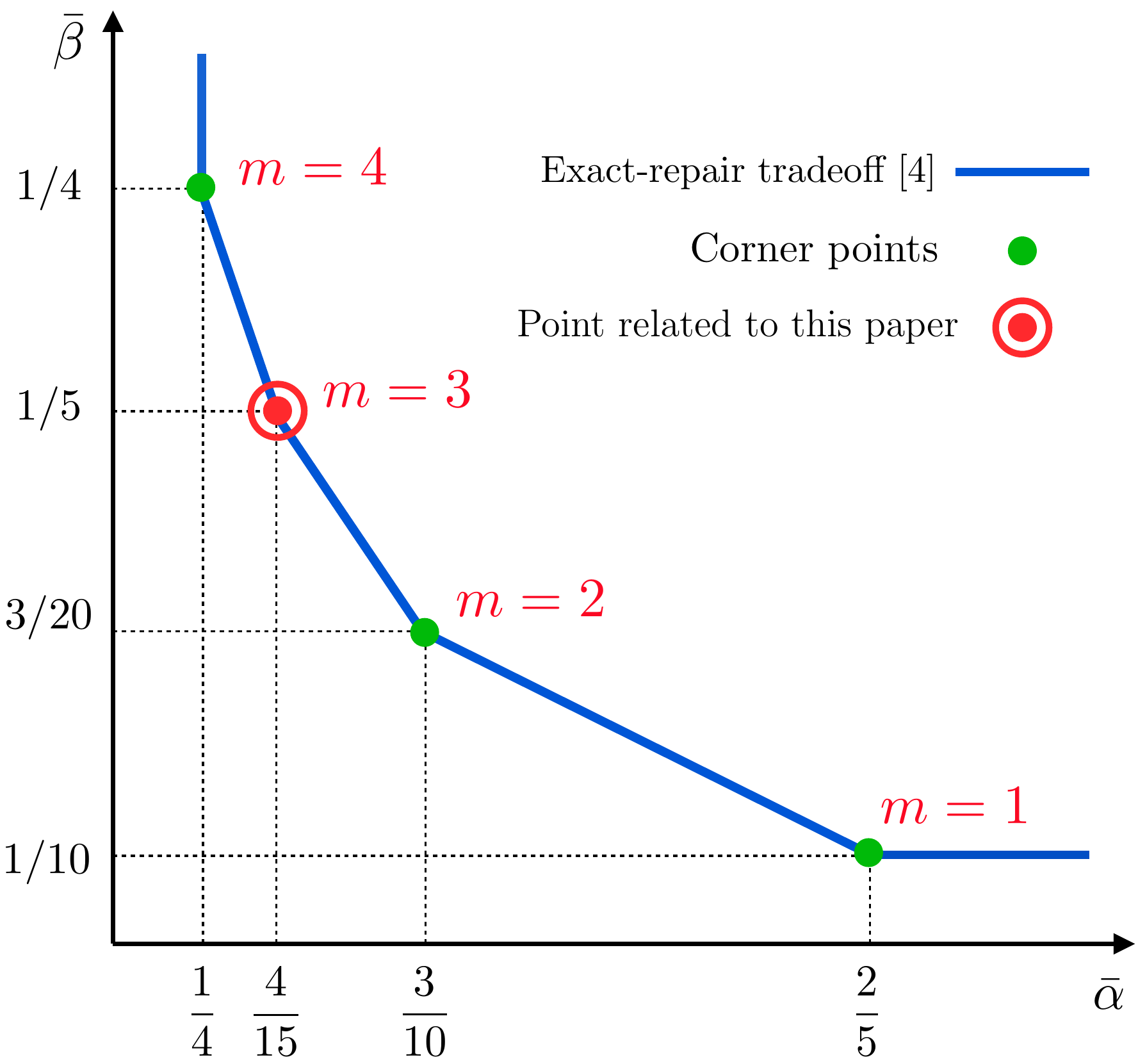}
\caption{{The exact-repair tradeoff for $(5,4,4)$-DSS.}} 
\label{fig:region}
\end{figure}

The two extreme points in \eqref{eq:corner}, namely $m=1$ and $m=k$,  correspond to the MBR and MSR points, respectively. Code constructions for MBR point are provided in \cite{rashmi2011optimal} for general $(n,k,d=k)$. Moreover, it is shown that these parameters 
are \emph{asymptotically} achievable for the MSR point \cite{suh2010existence, cadambe2010distributed}. However, it is not clear whether the interior of the tradeoff in \eqref{eq:corner} can be achieved when the  number of storage nodes is  larger than $k+1$, that is $n>k+1$. 

In this work we provide a partial answer to this question in a positive way: \textit{the corner points in \eqref{eq:corner} on the optimum tradeoff associated with $m=k-1$ is achievable}. We use a random coding strategy and probabilistic argument, to show that for any set of $n\geq k+1$ storage nodes forming $(n,k,k)$ exact-repair regeneration code, one can always \emph{find} a new node, such that the new node together with the existing $n$ nodes form an $(n+1,k,k)$ exact-repair regeneration code. Instead of constructing a new node, we rather pick it from an ensemble of nodes, and show that a random node preserves the desired properties with high probability, provided that the underlying field size is large enough. 

The main result of this work is formally stated in the following theorem. 

\begin{thm}
For any $(n,k,d=k)$ distribued storage system with $n>k$, there exist some large enough $q$, and exact repair regeneration codes over $\bbF_q$ with normalized per-node capacity and repair-bandwidth 
\[
(\bar \alpha,\bar \beta)=\left(\frac{m+1}{m(k+1)},\frac{m+1}{k(k+1)}\right)=\left(\frac{k}{k^2-1},\frac{k-1}{k^2-1}\right).
\] 
\label{thm:main}
\end{thm}

The performance of the code proposed in this paper for $k=4$ is compared to that of existing ones in the literature in Figure~\ref{fig:comp}. As it is shown, for both codes proposed in \cite{tian2014layered} and \cite{goparaju2014new}   the maximum file size can be stored in an ER regeneration system for given $(\alpha,\beta)=(k,k-1)$ decreases with $n$, whereas number of nodes is irrelevant to the file size in the new proposed code. 
 \begin{figure}[t]
	 	\hspace{-3mm}\centering
	 	\includegraphics[width=0.5\textwidth]{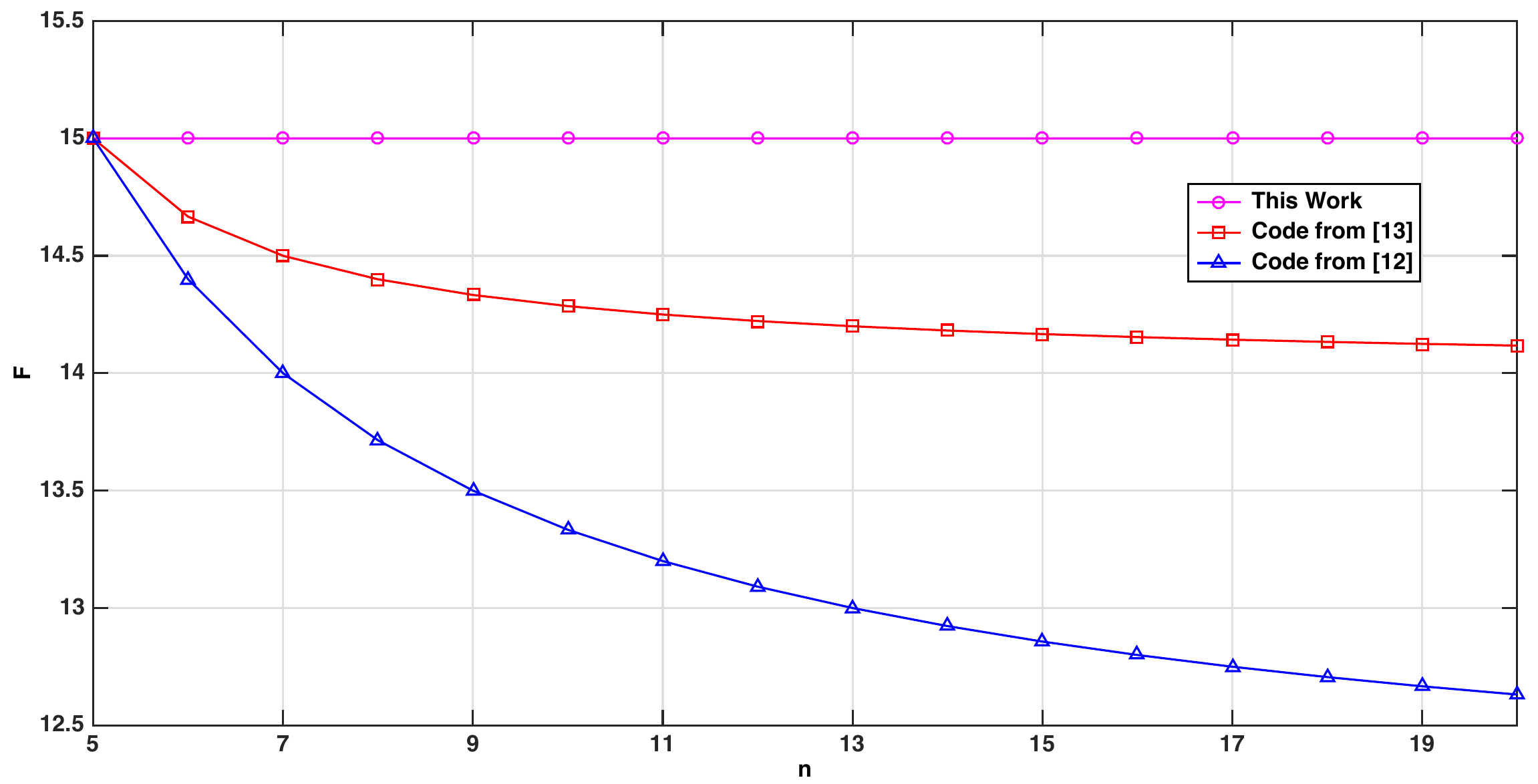}
	 	\caption{ER capacity of $(n,4,4)$-DSS as a function $n$ for $(\alpha,\beta)=(4,3)$.} 
	 	\label{fig:comp}
\end{figure}

Note that existence of such ER regeneration codes provides a partial characterization for the optimum tradeoff of an $(n,k,k)$-DSS as follows. 

\begin{cor}
The optimum tradeoff between the per-node capacity and repair-bandwidth of any  distributed storage system with parameters $(n,k,k)$ is given by 
\[
\bar \beta + (k-1)\bar \alpha \geq 1,
\]
for $(k-1)/k \leq \bar\alpha/\bar\beta\leq 1$. 
\label{cor:partial}
\end{cor}

As another consequence, Theorem~\ref{thm:main} suffices to characterize the entire optimum tradeoff for an $(n,3,3)$-DSS. 

\begin{cor}
The optimum tradeoff between the per-node capacity and repair-bandwidth of any exact-repair $(n,k=3,d=3)$ distributed storage system is given by
\begin{align*}
3\bar \alpha \geq 1, \qquad 
2\bar \alpha + \bar \beta \geq 1, \qquad 
4\bar \alpha + 6\bar \beta \geq 3, \qquad 
6\bar \beta \geq 1.
\end{align*}
\vspace{-5mm}
\label{cor:n33}
\end{cor} 
The proof of these results  can be found in Section~\ref{sec:proof}.

\section{Code Structure and Well-Aligned Node}
\label{sec:code-struct}

We start by redefining exact-repair regeneration codes in a linear framework. 

\subsection{Linear Codes}
It is easier to analyze linear regeneration codes in the context of vector spaces over a finite field.  In this context, we  denote subspaces by script letters, e.g. $\cW, \cS$ , and with slightly abuse of notation, use $\subseteq$ to denote subspace relationship. Moreover, we use $\sspan{\cdot}$ to denote the span of a set of vectors. Furthermore, we adopt the notation used in \cite{elyasi2015linear} as follows: 

A linear exact-repair regeneration code with parameters $(\alpha,\beta,F)$ for an $(n,k,d)$ distributed storage system can be defined as
\begin{itemize}
\item An $F$-dimensional vector space in $\bbF_q$ for some $q$, which we denote by $\cF$. The $F$ data symbols stored in the DSS form a basis for this vector space. 	

\item There are a total of $\alpha$ vectors from $\cF$ stored in node $i$, for $i\in \intv n$. These vectors span a subspace $\cW_i\subseteq \cF$, with $\dm{\cW_i}\leq \alpha$. 

\item \textit{Data recovery}: By accessing any subset of $k$ nodes, the entire vector space $\cF$ can be spanned. In other words, 
\[
\cF=\sum_{ j\in \cA}{\node j}\qquad  \forall \cA \subseteq \intv{n}, \size \cA =k.
\]
\item \textit{Node repair}: In the case of failure of node $x$, its content can be spanned by  summation of $d$ \emph{helping} vector spaces, each of dimensional $\beta$ and coming from one other one. More precisely, for every subset of nodes $\cA\subseteq \intv{n}\setminus\{x\}$ with $|\cA|=d$, we have\footnote{Note that the repair sent by node $j$ to $x$ potentially depends on the other helpers. This is captured by $[\cA]$ in our notation $\repairSpc{j}[x][\cA]$. However, we may drop $[\cA]$, whenever either $\cA$ is unique, or dependency of $\cA$ is clear from the context.}
\[
\node x \subseteq \sum_{j\in \cA} {\repairSpc{j}[x][\cA]}
\]
where ${\repairSpc{j}[x][\cA]}\subseteq \cW_j$ and $\dm{{\repairSpc{j}[x][\cA]}}\leq \beta$.
\end{itemize}

The following proposition lists some structural properties of optimum linear ER regeneration codes.
\begin{prop}
\label{prop:main}
Let $\cC$ be any $(n,k,k)$ ER regeneration code operating at $(\alpha, \beta, F)=(k,k-1,k^2-1)$. Consider the repair process of  node $x\in \intv n$ via nodes in $\cA$, where $\cA\subseteq \intv n \setminus \{x\}$, and $|\cA|=k$. Then vector space of each node $j \in \cA$ can be partitioned into two subspaces $\repairSpc{j}[x][\cA]$ and $\cmpSpc{j}[x][\cA]$ such that
\begin{enumerate}[(i)]
	 \item   $\repairSpc{j}[x][\cA]$ is the subspace node $j$ sends to repair node $x$, and $\dm{\repairSpc{j}[x][\cA]}=\beta=(k-1)$;\label{prop:main-1}\\[-3mm]
	\item   $\cmpSpc{j}[x][\cA]=\sspan {\cmpVec{j}[x][\cA]}$ is a one-dimensional subspace spanned by $\cmpVec{j}[x][\cA]\in \cW_j$;\label{prop:main-2}\\[-3mm]
	\item $\node j = \repairSpc{j}[x][\cA] \oplus \cmpSpc{j}[x][\cA]$;\label{prop:main-3}\\[-3mm]
	\item   $\sum_{j\in \cA} \cmpVec{j}[x][\cA]=0$;\label{prop:main-4}\\[-3mm]
	\item  $\cmpSpc{}[x][\cA]\triangleq \sum_{j\in \cA}\cmpSpc{j}[x][\cA]$ is a $(k-1)$-dimensional subspace of $\cF$;\label{prop:main-5}\\[-3mm]
	\item  $\cmpSpc{}[x][\cA]$ is spanned by $\{\cmpVec{j}[x][\cA]: j\in \cB\}$, for every $\cB\subseteq \cA$ with $|\cB|=k-1$.\label{prop:main-6}
\end{enumerate}
\end{prop} 

Partitioning of node spaces is given for a $(4,3,3)$ system in Figure~\ref{fig:533}, which can be useful to follow the statements above.  This proposition plays a central role in our arguments. The proof of the proposition can be found in Appendix~\ref{app:pr:prop:main}.

\subsection{Well-Aligned nodes}
\label{well-sec}
Our approach to achieve $(\alpha, \beta)$ for an $(n,k,k)$-DSS is recursive, i.e., we start with an existing $(k+1,k,k)$ ER code, and in each step append one new node to the system. Such new node should satisfy a set of conditions so that system maintain the data-recovery and exact-repair properties. It will be later shown that such conditions are fulfilled  by any node satisfy the following definition.

\begin{definition}
Let $\cC$ be an $(n,k,k)$ ER regeneration code with parameters $(\alpha,\beta,F)=(k,k-1,k^2-1)$. Fix  $\cA\subseteq \intv n$ with $|\cA|=k$, and let $x\notin \cA$ be a node index. An $\alpha$-dimensional subspace (node) $\node{\star}$ is called \emph{well-aligned with respect to a pair $(\cA;x)$} if it can be spanned by a set of basis vectors $\sB=\set{\nSym{1},\nSym{2},\dots,\nSym{k}}$, such that [see Figure~\ref{fig:well-aligned}]
\begin{enumerate}
\item[{\sf (1)}] $\nSym{i} = \sum_{j\in \cA}\sVec{i}{j} + \tau{\scriptstyle(i)}$\\ for  every $i=1,\dots, k$, where $\sVec{i}{j} \in \repairSpc{j}[x][\cA]$, and  $\tau{\scriptstyle(i)} \in \cmpSpc{}[x][\cA]$;\\[-3mm]

\item[{\sf (2)}] for each $j\in \cA$, there exists $i_j\in \intv k$ such that $\sVec{i_j}{j}=\vec 0$;\\[-3mm]

\item[{\sf (3)}] for every $j\in \cA$, set of vectors 
\[\set{\sVec{1}{j},\sVec{2}{j},,\dots,\sVec{k}{j}}\setminus \{\sVec{i_j}{j}\}\] 
are $(k-1)$ linearly independent vectors from $\repairSpc{j}[x][\cA]$.
\end{enumerate}
\label{def:well}
\end{definition}

 \begin{figure*}[t]
 	\centering
 	\includegraphics[width=0.89\textwidth]{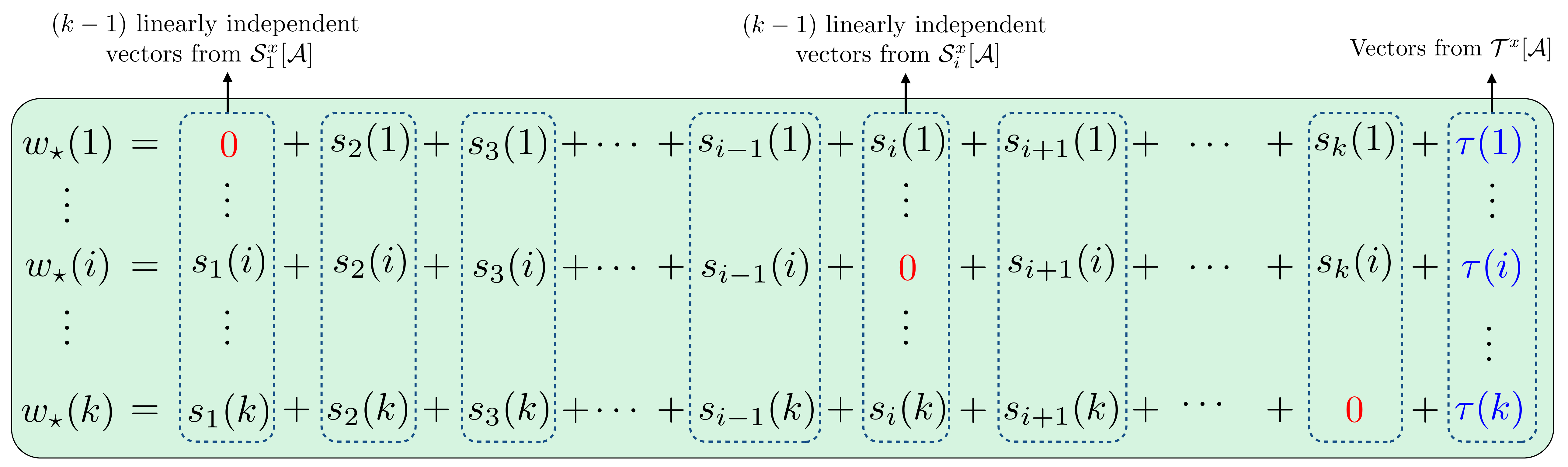}
 	\caption{Structure of a well-aligned node with respect to $\cA=\{1,2,\dots,k\}$.} 
	\label{fig:well-aligned}
 \end{figure*}

The next proposition shows that a well-aligned node w.r.t. $(\cA,x)$ satisfies some desired properties for an exact repair regeneration code. 

\begin{prop}
Let $\cC$ be an $(n,k,k)$ ER regeneration code with parameters $(\alpha,\beta,F)=(k,k-1,k^2-1)$. Consider $\cA$ to be a subset of $k$ nodes, and $x\in \intv n\setminus \cA$. Let $\node{\star}$ be a node that is well-aligned  w.r.t. $(\cA;x)$. Then $\cC^\star = \{\node i; i\in \cA\} \cup \{\node{\star}\}$ form a $(k+1,k,k)$ ER regeneration code with the same parameters $(\alpha,\beta,F)$.
\label{prop:well}
\end{prop}

\begin{proof}[Proof of Proposition~\ref{prop:well}]
	In order to prove this proposition we need to show  data-recovery and node repair properties of $\cC^*$.

\textbf{Node repair:} 
Note that $\cC^\star$ has only $(k+1)$ nodes, and once the failed node $x\in \cA \cup \{\star\}$ is chosen, the helper nodes are uniquely determined. To simplify the presentation, w.o.l.g. we may assume $\cA=\{1,2,\dots,k\}$ and $i_j=j$ for every $j\in \cA$, i.e., $\sVec{j}{j}=\vec{0}$. We can identify the following two cases. 
\begin{enumerate}
\item \textit{Repair of  $\node{\star}$ using nodes in $\cA$:} Note that $\node{\star}=\sspan{\sB}$. For this repair process,  we reconstruct each vector in $\sB$ using repair vectors sent from nodes in $\cA$. 
Recall by part \eqref{prop:main-6} of Proposition~\ref{prop:main} that since $\tau(i)\in\cmpSpc{}[x][]$, it can be  represented as
	\begin{equation*}
	 \label{tau-eq}
	 \tau(i)=\sum\nolimits_{j\in\cA\setminus\set{i}}\theta_{i,j}\cmpVec{j}[x][\cA]
	\end{equation*}
	for every $i\in \cA$. 
Now, we define the repair subspace $\repairSpc{j}[\star][\cA]$ sent by node $j\in \cA$ to node $\star$ by
\[
\repairSpc{j}[\star][\cA] = \sspan{\set{\sVec{i}{j}+\theta_{i,j}\cmpVec{j}[x][\cA]:j=1,2,\dots,k,j\neq i}}.
\]
It is clear that $\repairSpc{j}[\star][\cA]$ is spanned by only $(k-1)$ vectors, and hence its dimension does not exceed $\beta$. Now, vector $\nSym{i}\in \sB$ can be reconstructed via
\begin{align*}
\nSym{i} &= 
\sum_{j\in\cA\setminus\set{i}}\sVec{i}{j}+ \tau(i) \\
&= \sum_{j\in\cA\setminus\set{i}} 
\left(\sVec{i}{j}+ \theta_{i,j}\cmpVec{j}[x][\cA]\right) \in \sum_{j\in\cA\setminus\set{i}} \repairSpc{j}[*][\cA].
\end{align*}
	
\item \textit{Repair of node $\ell\in\cA$ using nodes in $\cB\triangleq \{\star\} \cup (\cA\setminus\set{\ell})$:} This repair process  is more technical. However, the illustration in Figure~\ref{fig:repair} can be helpful to follow the proof. Recall that $\node \ell = \repairSpc{\ell}[x][\cA] \oplus \sspan{\cmpVec{j}[x][\cA]}$. In order to repair $\node \ell$, we rebuild $\repairSpc{\ell}[x][\cA]$ and $\sspan{\cmpVec{j}[x][\cA]}$, separately. To reconstruct $\repairSpc{\ell}[x][\cA]$, node $\node{\star}$ sends a subspace
\[
\repairSpc{\star}[\ell][\cB] \triangleq \sspan{\{\nSym{i}: i\neq \ell\}}.
\] 
Note that $\repairSpc{\star}[\ell][\cB]$ is spanned by $(k-1)$ linearly independent vectors, thus $\dm{\repairSpc{\star}[\ell][\cB]}=k-1=\beta$. Moreover, each vector in 
$\{\nSym{i}: i\neq \ell\}$ can be written as 
\begin{align}
\nSym{i} = \sVec{i}{\ell} + \sum_{j\in \cA\setminus\{i,\ell\}} \sVec{i}{j} + \tau(i),
\label{eq:10}
\end{align}
which is indeed a vector from $\repairSpc{\ell}[x][\cA]$ that is corrupted by some interference. 

The repair data sent by other nodes in $\cA\setminus \{\ell\}$ play two important roles: {\sf (i)} cancel the interference in $\nSym{i}$'s, and {\sf (ii)} recover the vector space $\sspan{\cmpVec{\ell}[x][\cA]}$. Let us define 
\[
\repairSpc{j}[\ell][\cB] = \sspan{\set{\sVec{i}{j}: i\in \cA\setminus\set{j,\ell}}} \oplus 
\sspan{\cmpVec{j}[x][\cA]}.
\]
First, it is clear that $\repairSpc{j}[\ell][\cB] \subseteq \node j$. Moreover since $|\set{\sVec{i}{j}: i\in \cA\setminus\set{j,\ell}}|=k-2$, we have $\dm{\repairSpc{j}[\ell][\cB]}\leq (k-2) +1 = k-1 =\beta$. Now, from part \eqref{prop:main-6} of Proposition~\ref{prop:main} node $\ell$ can first recover $\cmpSpc{}[x][\cA]$ from the $(k-1)$ received vectors  $\{\cmpVec{j}[x][\cA]: j\in \cA, j\neq \ell\}$. Once $\cmpSpc{}[x][\cA]$ is rebuilt, the subspace $\cmpSpc{\ell}[x][\cA] \subseteq \cmpSpc{}[x][\cA]$ can be also recovered. 

Furthermore, since $\tau(i)\in \cmpSpc{}[x][\cA]$, they can be all reconstructed and canceled from $\nSym{i}$ in \eqref{eq:10}. The remaining interference in $\nSym{i}$, given by $\sum_{j\in \cA\setminus\{i,\ell\}} \sVec{i}{j}$, can be canceled since $\sVec{i}{j}\in \repairSpc{j}[\ell][\cB]$ for every $j\in \cA\setminus \{i,\ell\}$. By removing all interference,  vectors
\[ 
\{
\sVec{1}{\ell}, \dots, \sVec{\ell-1}{\ell}, \sVec{\ell+1}{\ell}, \dots, \sVec{k}{\ell}
\}
\]
can be reconstructed at the failed node. Then, since these vectors are linearly independent (by part {\sf (3)} of Definition~\ref{def:well}), they can completeley span $\repairSpc{\ell}[x][\cA]$, which together with $\cmpSpc{\ell}[x][\cA]$, can span $\node \ell$. 

\end{enumerate}
 \begin{figure*}[t]
 	\centering
 	\includegraphics[width=.85\textwidth,scale=1]{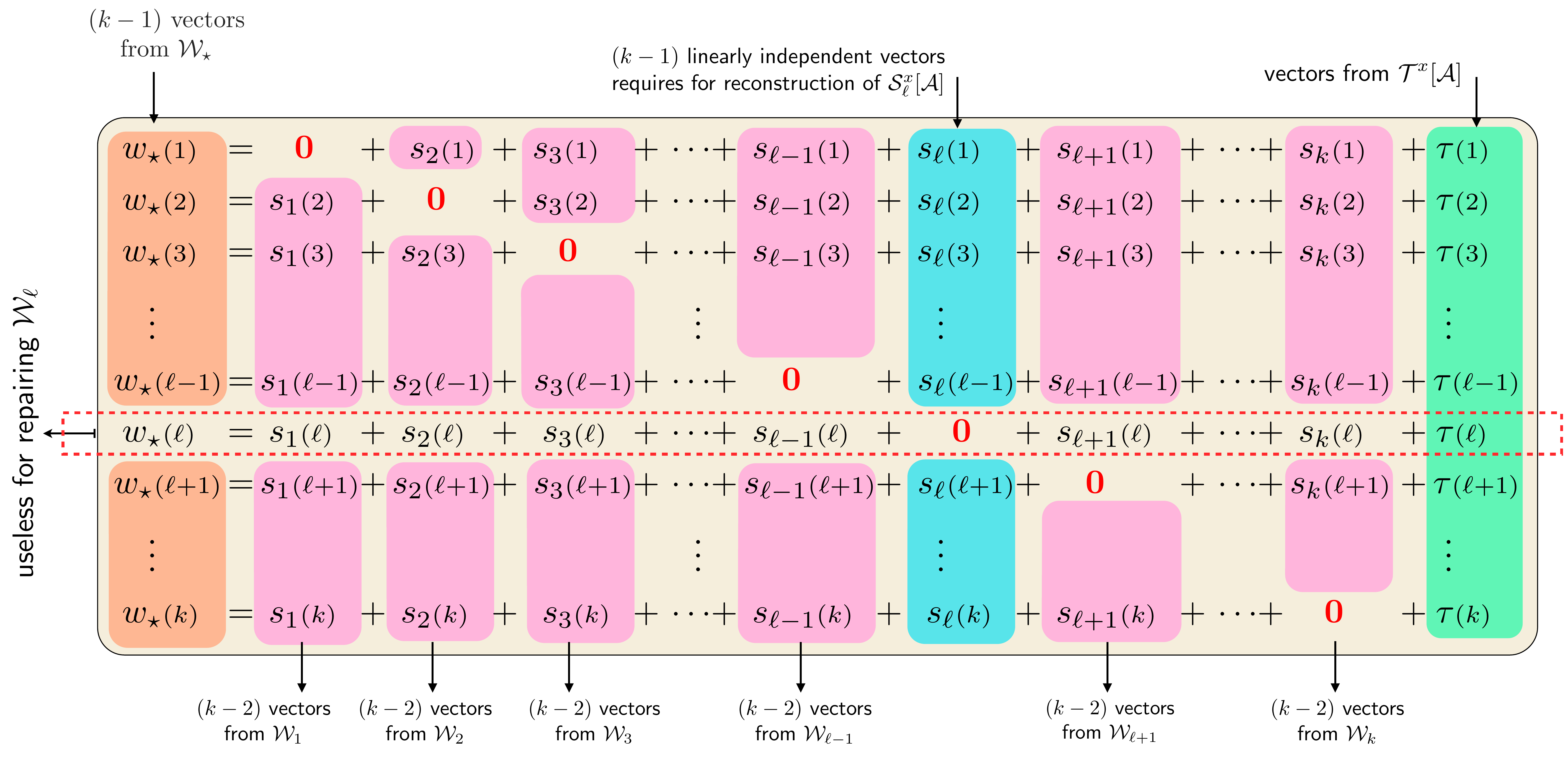}
 	\caption{Repair process of node $\ell$ via  $\{\star\} \cup \{1,2,\dots, \ell-1, \ell+1,\dots, k\}$ .}  	\vspace{-3mm}
 	\label{fig:repair}
 \end{figure*}

\textbf{Data recovery:} Recovering file from nodes in $\cA$ is clear, since nodes in $\cA$  were already a part of an $(n,k,k)$ ER code. Consider a set $\cB=\{*\} \cup \cA\setminus\set {\ell}$.  The argument used to proof repairability of $\node \ell$ shows that 
\[
\repairSpc{\ell}[x][\cA] \subseteq  \sum_{i\in \cB} \repairSpc{i}[\ell][\cB] \subseteq \sum_{i\in \cB} \node i.
\]
Moreover, part \eqref{prop:main-6} of Proposition~\ref{prop:main} implies 
\[
\cmpSpc{\ell}[x][\cA]\subseteq  \cmpSpc{}[x][\cA] = \sum_{i\in \cA\setminus \set{\ell}} \cmpSpc{i}[\ell][\cB] \subseteq \sum_{i\in \cA \setminus \{\ell\}} \node i.
\]
Hence, we have 
\begin{align*}
\cF &= \node \ell + \sum_{i\in \cA \setminus \{\ell\}} \node i
\subseteq (\repairSpc{\ell}[x][\cA] + \cmpSpc{\ell}[x][\cA]) + \sum_{i\in \cA \setminus \{\ell\}} \node i\\
&\subseteq \node{\star} + \sum_{i\in \cA \setminus \{\ell\}} \node i = \sum_{i\in \cB } \node i,
\end{align*}
which implies that $\cF$ can be recovered from nodes in $\cB$. This completes the proof of the data-recovery property. 
\end{proof}

\begin{remark}
Note that a subspace $\node{\star}$ can be  well aligned w.r.t. $(\cA;x)$ and not be well-aligned w.r.t. $(\cA;y)$ for $x\neq y$. However, nodes in $\cA$ together with $\node{\star}$ form a $(k+1,k,k)$ ER regeneration code as long as there exist at least one $x$ such that $\node{\star}$ is well-aligned w.r.t. $(\cA;x)$.
\end{remark}

\section{Probabilistic Method for Appending New Nodes}
\label{sec:prob}
Consider an $(n,k,k)$ ER regeneration code, $\cC=\{\node i: i\in \intv n\}$. Our goal is to append a new node $\node{\star}$ to this system, so that $\tilde{\cC} = \cC \cup \{\node{\star}\}$ maintain properties of ER regeneration codes. Such a code will be an $(n+1,k,k)$ ER code.  

It is easy to see that the necessary and sufficient condition for a new node $\node{\star}$ to be feasible to be added to $\tilde{\cC}$ is the following: $\tilde{\cC}$ is an $(n+1,k,k)$ ER regeneration code if and only if $\cC^\star_A =\{\node i: i\in \cA\} \cup \{\node{\star}\}$ form a $(k+1,k,k)$ ER regeneration code for every $\cA\subseteq \intv n$ with $|\cA|=k$.

Furthermore, Proposition~\ref{prop:well} shows that in order to $\cC^\star_A =\{\node i: i\in \cA\} \cup \{\node{\star}\}$ be a $(k+1,k,k)$ ER code, it suffices that $\node{\star}$ be well-aligned w.r.t. $(\cA,x)$ for some $x\in \intv n \setminus \cA$. 
	
In this section we will show that a randomly chosen $\alpha$-dimensional subspace $\node{\star}$ is well-aligned w.r.t. every choice of $(\cA;x)$ (with $|\cA|=k$) for some $x\in \cA$, with high probability, and hence can be add to the current code with $n$ nodes. More formally, we prove the following proposition.

\begin{prop}
Let $\cC=\{\cW_i: i\in \intv n\}$ be an $(n,k,k)$ ER regeneration code storing a data space $\cF$ over $\bbF_q$, and $\node{\star}$ be an $\alpha$-dimensional subspace of $\cF$ drawn uniformly at random. Then for any $\epsilon>0$, there exists a large enough $q$ such that 
\[
\P [\cC \cup \{\node{\star}\} \text{ is an } (n+1,k,k) \text{ ER code}] > 1-\epsilon.
\]
\label{prop:pr}
\vspace{-20pt}
\end{prop}
 The rest of this section is dedicated to prove Proposition~\ref{prop:pr}.
Let $E(\cA;x)$  be the event that a random $\node{\star}$ is well-aligned w.r.t. $(\cA;x)$ for $\cA\subseteq \intv n$ and $x\notin \cA$. Moreover, define
\[E(\cA)=\bigcup_{x\in\intv n \setminus \cA}E(\cA;x),\]
which is the event of existence of some $x$ for which $\node{\star}$ is well-aligned w.r.t. $(\cA,x)$.

Now let $\node{\star}$ be a subspace of $\cF$ picked uniformly at random among all $\alpha$-dimensional subspaces. Then we have
\begin{align}
\P  \left[ \bigcap_{\cA \subset \intv{n},\size{\cA}=k} E(\cA) \right] &= 1-\P \left[ \bigcup_{\cA \subset \intv{n},\size{\cA}=k} E^c(\cA) \right]\nonumber\\
&\geq 1-\sum_{\cA \subset \intv{n},\size{\cA}=k} \P \left[E^c(\cA) \right]\nonumber\\
&= 1-\sum_{\cA \subset \intv{n},\size{\cA}=k} (1-\P \left[E(\cA) \right])\nonumber\\
&= 1-\binom{n}{k} (1-\P \left[E(\cA_0) \right]),
\label{eq:pr-main}
\end{align}
where $\cA_0 =\{1,2,\dots,k\}$ is a fixed subset of $\intv n$. here, the last equality 
is due to the symmetric structure of nodes, which implies $\P[E(\cA)]$ does not  depend on the realization of $\cA$.

Next, let $x_0$ be a fixed node in $\intv{n}\setminus\cA_0$. We have 
$E(\cA_0;x_0) \subseteq E(\cA_0)$. Hence, 
\[
\P \left[E(\cA_0) \right]) \geq \P \left[E(\cA_0;x_0) \right]),
\]
Therefore, we can further lower bound RHS of \eqref{eq:pr-main} by 
	
	\begin{equation}
	\begin{split}
	\P  \left[ \bigcap_{\cA \subset \intv{n},\size{\cA}=k} E(\cA) \right] 
	&\geq 1-\binom{n}{k} (1-\P \left[E(\cA_0;x_0) \right]).
	\end{split}
	\label{pr-main2}
	\end{equation}
	
Thus, in order to show 
$
\P  \left[ \bigcap_{\cA \subset \intv{n},\size{\cA}=k} E(\cA) \right] \xrightarrow{q\rightarrow \infty} 1
$,  
it suffices to prove 
$\P \left[E(\cA_0;x_0)\right] \rightarrow 0$ as $q$ grows. 
To this end, we can evaluate $\P \left[E(\cA_0) \right]$ by
\begin{align}
\P \left[E(\cA_0;x_0) \right]=\frac{\#\text{ well-aligned subspaces w.r.t }(A_0;x_0)}{\# \;\alpha\text{-dimesnional subspaces of }\cF}.
\label{eq:ratio}
\end{align}

It is well-known (e.g. see \cite{knuth1971subspaces}) that the denominator in \eqref{eq:ratio} for $\dm{\cF}=F$ is given by 
\begin{equation}
	\frac{\prod_{h=0}^{\alpha-1}(q^F-q^h)}{\prod_{h=0}^{\alpha-1}(q^\alpha-q^h)}=\frac{\prod_{h=0}^{k-1}(q^{k^2-1}-q^h)}{\prod_{h=0}^{k-1}(q^k-q^h)}
	\label{eq:pr-1}	
\end{equation}
	
In order to compute the nominator, we must count the total number of well-aligned nodes w.r.t. a fixed pair $(\cA_0, x_0)$. Note that  for each  $i\in \cA$, we need to pick a total of $(k-1)$ vectors, namely $\{\sVec{j}{i}: j\in \cA\setminus\{i\}\}$, from a $(k-1)$-dimensional space $\repairSpc{i}[x_0][\cA]$ (see the $i$-th column  in Figure~\ref{fig:well-aligned}). For a given $i$, this can be done in $\prod_{h=0}^{k-2}(q^{k-1}-q^h)$ ways. 

On the other hand, vectors $\tau(j)$ in the last column are picked arbitrarily and independent of each other from an $(k-1)$-dimensional space $\cmpSpc{}[x][\cA]$. Therefore there are a total $q^{k-1}$ choices for each $\tau(i)$. 

Hence, the number of choices for the basis set of $\node{\star}$ is 
\begin{equation*}
	\label{pr-4}
	\prod_{i=1}^k \left[ \prod_{h=0}^{k-2}(q^{k-1}\hspace{-2pt}-q^h) \right] \cdot \prod_{j=1}^k q^{k-1} \hspace{-2pt} = \hspace{-2pt}
	\left[\prod_{h=0}^{k-2}(q^{k-1}\hspace{-2pt}-q^h)\right]^k \hspace{-4pt} q^{k(k-1)}.
\end{equation*}
Finally note that once vectors $\sVec{i}{j}$ and $\tau(j)$ are fixed for $i,j\in \intv k$,  each basis vector $w_\star(j)$ in Figure~\ref{fig:well-aligned} can be scaled by any non-zero $\xi\in \bbF_q$, while the resulting vector space is preserved. Considering this fact for $k$ basis vectors of $\node{\star}$, the nominator of  \eqref{eq:ratio} can be evaluated by

\begin{equation}
	\label{pr-5}
	\frac{\left[\prod_{h=0}^{k-2}(q^{k-1}-q^h)\right]^k.q^{k(k-1)}}{q^k}.
\end{equation}

Replacing \eqref{eq:pr-1} and \eqref{pr-5} in \eqref{eq:ratio}, we get 
\begin{align}
	\P \left[E(\cA_0;x_0)\right]
	&=\frac{\left[\prod\limits_{h=0}^{k-2}(q^{k-1}-q^h)\right]^k\cdot q^{k(k-1)}\cdot\prod\limits_{h=0}^{k-1}(q^k-q^h)}{\left[\prod\limits_{h=0}^{k-1}(q^{k^2-1}-q^h)\right] \cdot q^{k}}\nonumber\\
	&=\frac{q^{k^3}(1-o(1))}{q^{k^3}(1-o(1))}=1-o(1),
		\label{pr-6}
\end{align}
where $o(1)$ vanishes as $q\rightarrow \infty$. This completes the proof.

\section{Proof of the Main Results}
\label{sec:proof}
We have developed all the techniques need to prove the main results of Section~\ref{sec:result} in the previous sections.
\begin{proof}[Proof of Theorem~\ref{thm:main}]
We present the proof of the theorem using a recursive argument. We first note that for $(k+1,k,k)$ DSS,  exact-repair regeneration codes with parameters $(\alpha,\beta,F)=(k,k-1,k^2-1)$ are introduced independently in \cite{tian2014layered} and \cite{goparaju2014new}. We use them as the starting point of the recursive argument. 

Moreover, Proposition~\ref{prop:pr} implies that, for large enough $q$, a randomly sampled $\alpha$-dimensional subspace can be appended to an $(n,k,k)$ system to form an $(n+1,k,k)$ exact-repair regeneration code. Hence, we can start with $n=k+1$, and repeat picking random new nodes $\node{\star}$ and checking whether they satisfy the desired properties. By repeating this procedure we can get as many number of nodes needed while the entire system preserves the exact-repair property. 
\end{proof}

\begin{proof}[Proof of Corollary~\ref{cor:partial}]
The optimality of this tradeoff can be simply seen using a cut-set argument
\begin{align*}
F &= H(W_1,\dots, W_k) \\
&\leq H(W_1,\dots, W_{k-1}) + H(W_{k}| W_{1},\dots, W_{k-1})\\
&\leq (k-1)\alpha + \beta,
\end{align*}
which implies $\bar{\beta} \geq 1- (k-1) \bar{\alpha}$. Moreover, for $(k-1)/k \leq \bar{\alpha}/\bar{\beta} \leq 1$, the two extreme points of this bound are given by
\begin{align*}
(\bar{\alpha}, \bar{\beta}) = \left(\frac{1}{k},\frac{1}{k}\right) \qquad\text{and}
\qquad (\bar{\alpha}, \bar{\beta}) = \left(\frac{k}{k^2-1},\frac{k-1}{k^2-1}\right).
\end{align*}
Note that the first point is MSR, for which achievability is known for arbitrary value of $n$ \cite{suh2010existence, cadambe2010distributed}. Achievability of the second point is proved in Theorem~\ref{thm:main}. This completes the proof. 
\end{proof}

\begin{proof}[Proof of Corollary~\ref{cor:n33}]
Note that since an $(n,3,3)$ DSS includes a $(4,3,3)$ DSS, its exact-repair tradeoff cannot be lower than that of a $(4,3,3)$ DSS. On the other hand, the exact repair tradeoff for a $(4,3,3)$ system is characterized by Tian \cite{tian2013rate} as given in the corollary. 

In order to show that this tradeoff is indeed achivable, we can focus on the corner points, since the intermediate points can be achieved by space-sharing. There are three corner points for this region, namely, MBR, MSR, and one given by $(\bar{\alpha}, \bar{\beta}) = \left(\frac{k}{k^2-1},\frac{k-1}{k^2-1}\right)$. Achievability of the first two points is known for arbitrary $n$ \cite{rashmi2011optimal, suh2010existence, cadambe2010distributed}. For the middle point, however, Theorem~\ref{thm:main} guarantees existence of $n$ nodes maintaining data-recovery and exact-repair prperties, for large enough $q$. Hence, the entire boundary of the tradeoff is achievable for any value of $n\geq k+1$.   
\end{proof}

 \begin{figure*}[t]
	 	
	 	\centering
	 	\includegraphics[width=0.95\textwidth]{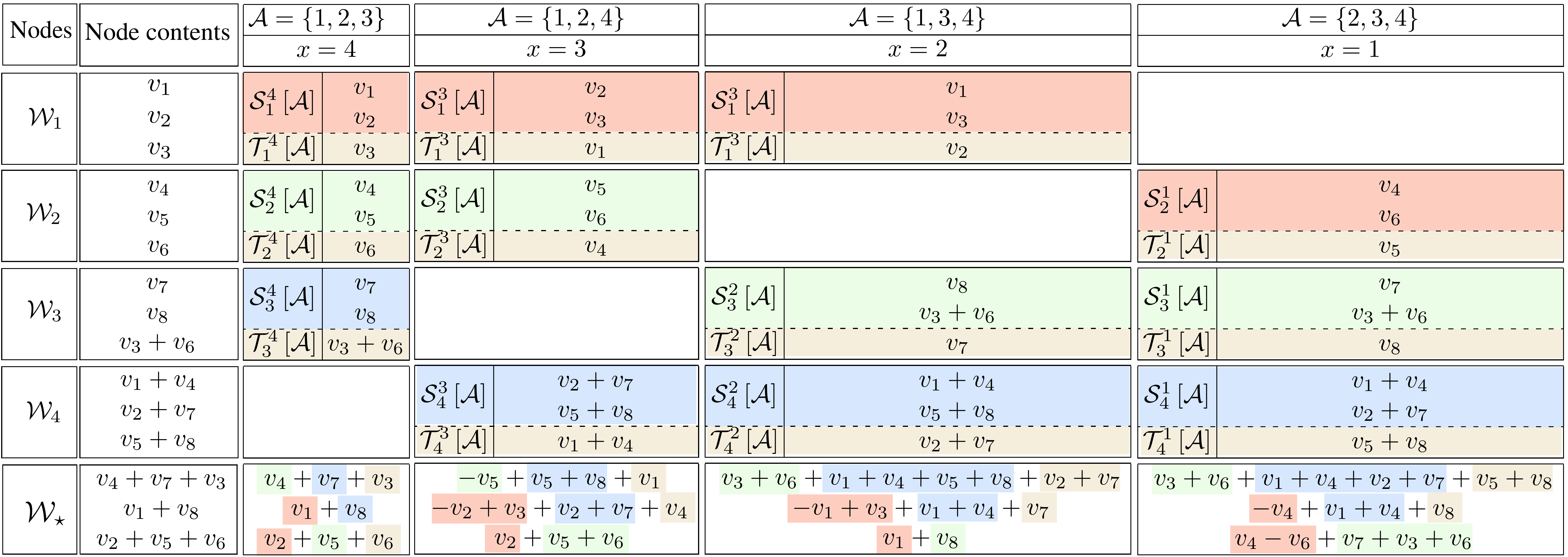}
	 	\caption{Construction of the $5$-th node in a $(5,3,3)$-DSS  based on a known $(4,3,3)$. The first demonstrates node contents, and the next columns show $\node{\star}$ is simultaneously well-aligned w.r.t. every $\cA$ with $|\cA|=3$.} \label{fig:533}
	 	\vspace{-2mm}
\end{figure*}

\appendices
\section{Proof of proposition~\ref{prop:main}}
\label{app:pr:prop:main}
The proof of this proposition is based analysis of the repair subspaces sent by nodes in $\cA$ to repair node $x$, namely $\repairSpc{j}[x][\cA]$ for $j\in \cA$. In the rest of this section, since $x$ and $\cA$ are fixed, we may drop superscript $x$ and parameter $\cA$ and write $\cS_j\triangleq \repairSpc{j}[x][\cA]$ for the sake of simplicity. 

Before we present the proof of Proposition~\ref{prop:main}, we state and prove a few more properties of the subspaces sent by nodes in the repair process of another node. 

\begin{prop}
	\label{repair-dep}
	Let $\cC$ be a linear exact-repair regenerating code operating at an optimum point $(\alpha,\beta,F)$. The repair subspaces sent by nodes in $\cA$ in order to repair node $x$ are mutually linearly independent. That is, if 
$\sum_{j\in \cA}\vec{v}_j = \vec{0}$, 
	holds for some $\vec{v}_j\in \cS_j$, then $\vec{v}_j=\vec{0}$ for every $j\in \cA$. 
\end{prop}
	
\begin{proof}[Proof of Proposition~\ref{repair-dep}]
We prove this by contradiction. Assume there exist vectors $\vec{v}_j\in \cS_j$ with at least one non-zero vector (say $\vec{v}_{\ell} \neq \vec{0}$) that sum up to zero. We have 
\begin{align}
\vec{v}_\ell = \sum_{j\in \cA\setminus\{\ell\}} (-\vec{v}_j) \in \sum_{j\in \cA\setminus\{\ell\}} \cS_j.
\label{eq:1}
\end{align}
Since $\vec{0}\neq \vec{v}_\ell \in \repairSpc{\ell}[x][\cA]$ and $\dm{\cS_{\ell}}=\beta$, there exist a  subspace $\hat{S}_\ell \subset \cS_\ell$ such that $\cS_\ell =  \hat{S}_\ell \oplus \sspan{\vec{v}_\ell} $ and $\dm{\hat{\cS_{\ell}}}\leq \beta-1 = k-2$. 
Now, from \eqref{eq:1} we have
\begin{align}
\cS_\ell =  \hat{S}_\ell \oplus \sspan{\vec{v}_{\ell}} \subset \hat{S}_\ell + \sum_{j\in \cA\setminus\{\ell\}} \cS_j.
\label{eq:2}
\end{align}
Next, note that  $|(\cA \setminus\{\ell\}) \cup \{x\}|=k$. Therefore, we have
\begin{align}
\cF &= \sum_{j\in (\cA \setminus \{\ell\}) \cup \{x\}} \node j = \node x + \sum_{\cA \setminus\{\ell\} } \node j \nonumber\\
&\subseteq \sum_{j\in \cA} \cS_j + \sum_{j\in \cA \setminus\{\ell\} } \node j \nonumber\\
&= \cS_\ell + \sum_{j\in \cA\setminus\{\ell\}} \cS_j+ \sum_{j\in \cA \setminus\{\ell\} } \node j \nonumber\\ 
&\stackrel{(a)}{\subseteq} \hat{\cS}_\ell + \sum_{j\in \cA\setminus\{\ell\}} \cS_j+ \sum_{\cA \setminus\{\ell\} } \node j \nonumber\\
&\stackrel{(b)}{=} \hat{\cS}_\ell + \sum_{j\in \cA \setminus\{\ell\} } \node j
\label{eq:3}
\end{align}
where $(a)$ is implied by \eqref{eq:2}, and $(b)$ holds since $\cS_j \subseteq \node j$. 
From \eqref{eq:3} we have
\begin{align*}
k^2-1 &= \dm{\cF} \leq \dm{\hat{\cS}_\ell} + \sum\nolimits_{\cA \setminus\{\ell\} } \dm{\node j} \\
&= (k-2) + (k-1)k = k^2-2,
\end{align*}
which is in infeasible. Thus the initial assumption is wrong, and repair subspaces are mutually linearly independent. 
\end{proof}

\begin{prop}
For a fixed pair $(\cA,x)$, define $\repairSpc{}=\sum_{i \in \cA} \repairSpc{i}[x][\cA]=\sum_{i \in \cA} \repairSpc{i}$. Then we have 
$
\node i \setminus \repairSpc{} \neq \varnothing
$ 
for every $ i\in \cA$.
\label{prop:W-S}
\end{prop}
\begin{proof}[Proof of Prposition~\ref{prop:W-S}]
We again prove this claim by contradiction. Suppose  there exists a node $\ell\in \cA$ such that $\node \ell \subseteq \repairSpc{}$. 
Now, fix a node $j\in \cA$ with $j \neq \ell$. We have	
\begin{align}
		\node \ell \subseteq \repairSpc{}= \sum_{i\in \cA} \repairSpc{i} &=\repairSpc{\ell}+\repairSpc{j}+ \sum_{ \scriptscriptstyle{ i\in \cA\setminus \set{j,\ell}} } \hspace{-10pt} \repairSpc{i}\nonumber\\
		&\subseteq \repairSpc{\ell}+\repairSpc{j}+ \sum_{ \scriptscriptstyle{ i\in \cA\setminus \set{j,\ell}} }  \node i.
\label{eq:4}		 
\end{align}
	
Next, since $|\{x\} \cup \cA\setminus\{ j\}|=k$, we have
\begin{align*}
\cF &= \sum_{ i\in \{x\} \cup \cA\setminus\{j\} } \node i 
=\node x + \node \ell + \sum_{ i\in \cA\setminus \set{j,\ell} } \node i \\
&\stackrel{(c)}{\subseteq} \left(\sum_{ i\in \cA } \repairSpc{i}\right)  
\hspace{-2pt}+\hspace{-2pt}
\left(\repairSpc{\ell}+\repairSpc{j}+ \hspace{-5pt}\sum_{ \scriptscriptstyle{ i\in \cA\setminus \set{j,\ell}} } \hspace{-7pt} \node i\right)
+ \hspace{-7pt}
\sum_{ i\in \cA\setminus \set{j,\ell} } \hspace{-7pt}\node i \\
&= \repairSpc{\ell}+\repairSpc{j} + \hspace{-3pt}\sum_{ \scriptscriptstyle{ i\in \cA\setminus \set{j,\ell}} }  \left(\node i + \repairSpc{i}\right) 
\stackrel{(d)}{=} \repairSpc{\ell}+\repairSpc{j} +  \hspace{-3pt}\sum_{ \scriptscriptstyle{ i\in \cA\setminus \set{j,\ell}} }  \node i 
\end{align*}
where we used \eqref{eq:4} in $(c)$, and equality in $(d)$ follows the fact that $\repairSpc{i} \subseteq \node i$. Therefore,
\begin{align*}
k^2-1=\dm{\cF} & 
 \leq \dm{\repairSpc{\ell}} +\dm{\repairSpc{j}} + \hspace{-4pt} \sum_{ \scriptscriptstyle{ i\in \cA\setminus \set{j,\ell}} } \hspace{-4pt}\dm{\node i}\\
&= 2(k-1) + (k-2)k= k^2-2,
\end{align*}
which is infeasible.  This implies our initial assumption is not true, and therefore the claim of the proposition holds. 
\end{proof}

Now, we are ready to  prove Proposition~\ref{prop:main}.

\begin{proof}[Proof of Proposition~\ref{prop:main}]
Consider the repair  of  $x$ by the help of  $\cA$, where each $j\in \cA$ sends  $\repairSpc{j}=\repairSpc{j}[x][\cA]=\subseteq \node j$. Since the code operates at $(\alpha,\beta,F)=(k,k-1,k^2-1)$, we have $\dm {\repairSpc{j}[x][\cA]}=k-1$. This proves part \label{prop:main-1} of the proposition. 
	
Recall that $\repairSpc{j}=\repairSpc{j}[x][\cA] \subseteq \node j$. We have $\dm{\node j} =k$, and $\dm{\repairSpc{j}} = k-1$. Hence, $\repairSpc{j}$ can be extended by a one-dimensional subspace $\nodeCmp{j} = \nodeCmp{j}[x][\cA]\triangleq\sspan{\ncmpVec{j}[x][\cA]} = \sspan{\vec{u}_j}$ to span the entire  space $\node j$ for every $j\in \cA$. In other words, there exist $\ncmpVec{j}[x][\cA]\in \node j$ such that $\node j =  \repairSpc{j}[x][\cA] \oplus \sspan{\ncmpVec{j}[x][\cA]}$. Note that in general there are many choices for $\sspan{\ncmpVec{j}[x][\cA]}$, and we may continue the proof with \emph{any} of them.

Next, note that there are a total of $k^2$ vectors stored in $\{\node j: j\in \cA\}$, and they all lie in $\cF$ which is an $(k^2-1)$ space. Hence, they cannot be all linearly independent. More precisely, there must be a set of  vectors one from each $\cW_j$, that sum up to zero. Since each vector in $\node i$ can be written as $\xi_j \vec{u}_j+\vec{s_j}$ for some $\xi_j \in \bbF_q$, and some $\vec{s}_j\in \repairSpc{j}$, we have
\begin{align}
\exists \  \vec{s_j}\in \repairSpc{j},\xi_j \in \bbF_q, j\in \cA :\ \ 
&\sum_{j\in \cA}(\xi_j \vec{u}_{j}+\vec{s_j}) =\vec{0}.
\label{eq:5}
\end{align}

We define $\cmpVec{j}[x][\cA] \triangleq \xi_j \vec{u}_{j}+\vec{s_j} \in \cW_j$ for every $j\in \cA$. 
Hence, part \eqref{prop:main-4} of the proposition is immediately implied by \eqref{eq:5}.
In order to prove part \eqref{prop:main-2} of the proposition, it suffices to show that $\xi_j \neq 0$ for $j\in \cA$. First note that there is at least one non-zero $\xi_j$, because otherwise \eqref{eq:5} implies existence of non-zero $\vec{s}_j\in \repairSpc{j}$  for $j\in \cA$ with $\sum_{j\in \cA} \vec{s}_j =\vec{0}$, which is in contradiction with Proposition~\ref{repair-dep}. So, let $\xi_i\neq 0$. Now if $\xi_\ell=0$ for some $\ell \in \cA$, we have
\begin{align*}
\ncmpVec{i}=  \sum_{j\in \cA\setminus\set{\ell,i}} -\xi_i^{-1}(\xi_j \ncmpVec{j}+\vec{s}_j)-\xi_i^{-1}\vec{s}_i-\xi_i^{-1}\vec{s}_\ell
\end{align*}
which implies 
\begin{align}
\nodeCmp{i} \subseteq \sum_{j\in \cA\setminus\set{\ell,j}} \node j + \repairSpc i + \repairSpc \ell 
\label{eq:6}
\end{align}

Since $|\{x\} \cup (\cA\setminus\{\ell\})|=k$, we have 
\begin{align}
\cF &= \cW_x + \sum_{j\in \cA \setminus \set \ell}{\node j} = 
\node x +\node i + \sum_{i\in \cA \setminus \set{\ell,i}}{\node j}\nonumber\\
&= \node x +\left( \nodeCmp{i} +\repairSpc{i} \right) + \sum_{i\in \cA \setminus \set{\ell,i}}{\node j}\nonumber\\
&\stackrel{(e)}{\subseteq} \hspace{-2pt}\left( \sum_{j\in \cA} \repairSpc{j}\right) \hspace{-2pt}+\hspace{-2pt} 
\left( \sum_{j\in \cA\setminus\set{\ell,j}} \hspace{-2pt}\node j + \repairSpc i + \repairSpc \ell  \right)\hspace{-2pt} +\hspace{-2pt}
\sum_{j\in \cA \setminus \set{\ell,i} }\hspace{-2pt}{\node j}\nonumber\\
&= \repairSpc i + \repairSpc \ell  + \sum_{j\in \cA \setminus \set{\ell,i} }{\node j},
\label{eq:7}
\end{align}
where $(e)$ is implied by \eqref{eq:6}. Therefore, from \eqref{eq:7} we have
\begin{align*}
k^2-1&=\dim{\cF} 
\leq \dm{\repairSpc i  + \repairSpc \ell+\hspace{-7pt} \sum_{j\in \cA \setminus \set{\ell,i} }\hspace{-2pt}{\node j}}\hspace{-2pt} \leq k^2-2.
\end{align*}
This implies 
\begin{align}
\xi_j\neq 0, \qquad \forall j\in \cA,
\label{eq:8}
\end{align} 	
which  completes the proof of part (ii) of the proposition.  	

Moreover, from \eqref{eq:8}, it is clear that $\vec{0} \neq \cmpVec{j}[x][\cA]\notin \cS_j$. Then, since  $\dm{\node j} = \dm{\repairSpc{j}[x][\cA]}+1$, one can conclude that $\node j = \repairSpc{j}[x][\cA] + \sspan{\cmpVec{j}[x][\cA]}$, which proves part \eqref{prop:main-3} of the proposition. 

In order to show part \eqref{prop:main-5}, we note that $\cmpSpc{}[x][\cA]=\sum_{j\in \cA} \cmpSpc{j}[x][\cA] = \sspan{\cmpVec{j}[x][\cA]; j\in A}$ is spanned by $k$ vectors. Moreover, part \eqref{prop:main-4} of the proposition implies that these $k$ vectors are not linearly independent, and hence $\dm{\cmpVec{}[x][\cA]} \leq k-1$.  Additionally, from  part  \eqref{prop:main-3} we have
\begin{align*}
\cF &= \sum_{j\in \cA} \node j = \sum_{j\in \cA} \left(\repairSpc{j}[x][\cA] + \cmpSpc{j}[x][\cA]\right)\\
&= \sum_{j\in \cA} \repairSpc{j}[x][\cA] + \sum_{j\in \cA} \cmpSpc{j}[x][\cA]
= \sum_{j\in \cA} \repairSpc{j}[x][\cA] +  \cmpSpc{}[x][\cA]
\end{align*}
Hence, 
$
\dm{\cF} \leq \sum_{j\in \cA} \dm{\repairSpc{j}[x][\cA]} + \dm{ \cmpSpc{}[x][\cA]}
$, 
and so $\dm{ \cmpSpc{}[x][\cA]} \geq (k^2-1) - k(k-1)=k-1$. This together with $\dm{\cmpSpc{}[x][\cA]} \leq k-1$ implies part \eqref{prop:main-5} of the proposition. 

Consider $\cB=\cA\setminus \{\ell\}$ for some $\ell\in \cA$.  
Note  from part \eqref{prop:main-4} of the proposition that $\cmpVec{\ell}[x][\cA] = -\sum_{j\in \cB} \cmpVec{j}[x][\cA]$, which implies $\cmpSpc{\ell}[x][\cA] \subseteq \sum_{j\in \cB} \cmpSpc{j}[x][\cA]$. Therefore, 
\begin{align*}
\cmpSpc{}[x][\cA] &= \sum_{j\in \cA} \cmpSpc{j}[x][\cA] = \cmpSpc{\ell}[x][\cA] + \sum_{j\in \cB} \cmpSpc{j}[x][\cA]  
= \sum_{j\in \cB} \cmpSpc{j}[x][\cA], 
\end{align*}
which means $\cmpSpc{}[x][\cA]$ is spanned by any $(k-1)$ subspaces of form $\cmpSpc{j}[x][\cA]$. This completes the proof of \eqref{prop:main-6}.
\end{proof}

\bibliography{Pr-ref}

\end{document}